\documentclass{aims}
\usepackage{amsmath}
\usepackage{amssymb}
\usepackage{eucal}
\usepackage{paralist}
\usepackage{enumerate}
\usepackage[colorlinks=true]{hyperref}
\hypersetup{urlcolor=blue, citecolor=red}
\def\currentvolume{X}
\def\currentissue{X}
\def\currentyear{200X}
\def\currentmonth{XX}
\def\ppages{X--XX}

\newtheorem{theorem}{Theorem}[section]
\newtheorem{corollary}{Corollary}[section]

\newtheorem{lemma}[theorem]{Lemma}
\newtheorem{proposition}{Proposition}[section]

\theoremstyle{definition}
\newtheorem{definition}[theorem]{Definition}

\newtheorem{example}{Example}[section]

\newcommand{\A}{{\mathcal A}}
\newcommand{\B}{{\mathcal B}}
\newcommand{\C}{{\mathcal C}}
\newcommand{\E}{{\mathcal E}}
\newcommand{\G}{{\mathcal G}}
\newcommand{\M}{{\mathcal M}}
\renewcommand{\S}{{\mathcal S}}

\title[Optimal Splitting Authentication Codes]
      {Infinite Families of Optimal Splitting Authentication Codes
      Secure Against Spoofing Attacks of Higher Order}

\author[Y. M. Chee, X. Zhang and H. Zhang]{}


 \email{ymchee@ntu.edu.sg}
 \email{xiandezhang@ntu.edu.sg}
 \email{z\_h1984@126.com}

\thanks{Y. M. Chee and X. Zhang are supported in part by the National Research
Foundation of Singapore under Research Grant NRF-CRP2-2007-03. Y. M. Chee is also
supported in part by the Nanyang Technological University under Research
Grant M58110040.}

\begin{document}
\maketitle

\centerline{\scshape Yeow Meng Chee and Xiande Zhang }
\medskip
{\footnotesize
 \centerline{Division of Mathematical Sciences}
   \centerline{School of Physical \& Mathematical Sciences}
   \centerline{Nanyang Technological University}
   \centerline{21 Nanyang Link, Singapore 637371}
} 

\medskip

\centerline{\scshape Hui Zhang}
\medskip
{\footnotesize
 \centerline{Department of Mathematics}
   \centerline{Zhejiang University}
   \centerline{Hangzhou 310027, Zhejiang, China}
}

\bigskip

 \centerline{(Communicated by Iwan Duursma)}

\begin{abstract}
We consider the problem of constructing optimal authentication codes with splitting.
New infinite families of such codes
are obtained. In particular, we establish the first known infinite family of optimal authentication
codes with splitting that are secure against spoofing attacks of order two.
\end{abstract}

\section{Introduction}

In the standard model of authentication theory
\cite{Simmons:1982,Simmons:1984,Simmons:1985,Simmons:1992},
a {\em transmitter} wants to send some information to a {\em receiver} across
an insecure channel
while an {\em opponent} with access to the channel wants to deceive the receiver. The opponent
can either insert new messages into the channel, or intercept messages from the transmitter
and modify them into his own. In each case, the opponent's goal is to deceive the receiver into
believing that the new messages are authentic (coming from the transmitter). The first attack
based on insertion of new messages is known as {\em impersonation} and the second attack
based on modification of messages from the transmitter is known as {\em substitution}.

More formally, let $\S$ denote the set of all {\em source states},
$\M$ be the set of all {\em messages},
and $\E$ be the set of all {\em encoding rules}. All these are finite sets.
A source state is the information the transmitter wishes to communicate to the receiver.
An encoding rule is an injection from $\S$ to $2^\M$.
The transmitter and receiver agree beforehand on a secret encoding rule $e\in \E$.
To communicate a source state $s\in\S$, the transmitter determines $M=e(s)$
(note that $M\subseteq \M$) and
chooses a message $m\in M$ to send to the receiver.
The receiver accepts the received message as authentic if there exists an $M$ in
the image of $e$ containing the received message.
For the receiver to recover the source
state, each encoding rule must satisfy the condition
\begin{equation*}
e(s) \cap e(s') = \varnothing, \ \ \text{for distinct $s,s'\in\S$.}
\end{equation*}
The triple $(\S,\M,\E)$ is called an {\em authentication code}, or {\em A-code} in short.

An A-code $(\S,\M,\E)$ can be represented by an $|\E|\times |\S|$ matrix,
whose rows are indexed by authentication rules, and columns indexed by source states, such
that the entry in row $e\in\E$ and column $s\in\S$ is $e(s)$.

For $k$ an integer
and $X$ a finite set, we denote by $\binom{X}{k}$ the set of all $k$-subsets of $X$.
Research
on authentication codes have focused on the case when every encoding rule
is an injection from $\S$ to $\binom{\M}{c}$, for some positive $c$. Such an A-code is called
{\em a $c$-splitting A-code}. A $1$-splitting A-code is also known as an
{\em A-code without splitting}, and a $c$-splitting A-code with $c\geq 2$
is known as an {\em A-code with splitting}. A-codes with splitting are useful for
the analysis of authentication with arbitration \cite{KurosawaObana:2001}, 
an extended model of authentication
introduced by Simmons \cite{Simmons:1987,Simmons:1990}
for the scenario when the transmitter and receiver may both
be deceptive.

In a {\em spoofing attack of order $i$} \cite{Massey:1986}, 
the opponent observes $i$ distinct messages
sent by the transmitter through the insecure channel under the same encoding rule.
The opponent then inserts a new message (distinct from the $i$ messages already sent),
hoping to have it accepted by the receiver as authentic. Within this framework, impersonation
and substitution attacks are just spoofing attacks of order zero and one, respectively.
While these attacks have been rather well studied for A-codes, less is known
for the case of spoofing attacks of order $i\geq 2$,
especially on $c$-splitting A-codes when $c\geq 2$.

The probability distribution on the set of source states $\S$ induces a probability distribution on
$\binom{\S}{i}$, $i\geq 0$. Given these probability distributions, the transmitter and receiver
choose a probability distribution on $\E$, called an {\em encoding strategy}. 
For any $s\in\S$ and $e\in \E$, the transmitter
also chooses a probability distribution on $e(s)$, called a {\em splitting strategy}.
The opponent
is assumed to know the encoding and splitting strategies.
The transmitter and receiver chooses the encoding and splitting strategies to minimize
the probability of being deceived by the opponent. We denote by $P_{d_i}$
the probability that the opponent can deceive the receiver with a spoofing attack
of order $i$. The following lower bound on $P_{d_i}$ is known.

\begin{proposition}[Huber \cite{Huber:2010}]
\label{HuberLB}
In a $c$-splitting A-code $(\S,\M,\E)$, 
\begin{equation*}
P_{d_i} \geq c \cdot  \frac{|\S|-i}{|\M|-i},
\end{equation*}
for every $i\geq 0$.
\end{proposition}

A $c$-splitting A-code is said to
be {\em $(t-1)$-fold secure against spoofing} if 
$P_{d_i}=c(|\S|-i)/(|\M|-i)$, for all $i$, $0\leq i<t$.
For succinctness, we call such a code a {\em $(t,c)$-splitting
A-code}.

Huber \cite{Huber:2010} also showed that the number of encoding rules must be large enough
if an A-code is to be $(t-1)$-fold secure against spoofing.

\begin{proposition}[Huber \cite{Huber:2010}]
\label{HuberE}
In a $(t,c)$-splitting A-code $(\S,\M,\E)$, 
\begin{equation*}
|\E| \geq \frac{1}{c^{t}}\cdot \frac{\binom{|\M|}{t}}{\binom{|\S|}{t}}.
\end{equation*}
\end{proposition}

For efficiency, we want the number of encoding rules in an A-code
to be as small as possible. We call a $(t,c)$-splitting A-code 
{\em optimal} if the lower bound in Proposition \ref{HuberE}
is met with equality.

The main contribution of this paper is
on the construction of optimal $(t,c)$-splitting A-codes with three
source states, 
for $c\geq 2$ and $t\in\{2,3\}$. In particular, we show that the following two new families
of A-codes exist:
\begin{enumerate}[(i)]
\item $(2,5)$-splitting A-codes with three source states and $v$ messages, for all
$v\equiv 1\bmod{150}$, $v\not=301$.
\item $(3,2)$-splitting A-codes with three source states and $v$ messages, for all
$v\equiv 2\bmod{8}$.
\end{enumerate}
The $(3,2)$-splitting A-codes we obtained is the first known
infinite family of $(t,c)$-splitting A-codes
with $t>2$ and $c>1$.
We also prove that a $(2,c)$-splitting A-code with $k$ source states and $v$
messages exists for all sufficiently large $v$ (with $k$ and $c$ fixed).

\section{Preliminaries}

This section serves to provide notions and
results that are required for our construction in subsequent
sections.

The ring $\mathbb{Z}/n\mathbb{Z}$ is denoted $\mathbb{Z}_n$.


\subsection{Design-Theoretic Background}

Huber \cite{Huber:2010} defined {\em splitting $t$-designs},
generalizing the splitting 2-designs of
Ogata {\em et al.} \cite{Ogataetal:2004}.

\begin{definition}
Let $t$, $v$, $k$, $c$, and $\lambda$ be positive integers, with $t\leq k$ and $ck\leq v$.
A {\em splitting $t$-design}, or more precisely,
a {\em splitting $t$-$(v,k\times c,\lambda)$ design}, is a pair $(X,\A)$ such that
\begin{enumerate}[(i)]
\item $X$ is a set of $v$ elements, called {\em points};
\item $\A$ is a set of $k\times c$ arrays, called {\em blocks}, with entries from $X$,
such that each point of $X$ occurs at most once in each block;
\item for every $\{x_i : 1\leq i\leq t\}\in\binom{X}{t}$, there are exactly $\lambda$ blocks
in which $x_i$, $1\leq i\leq t$, occur in $t$ different rows.
\end{enumerate}
\end{definition}

Note that a splitting $t$-$(v,k\times 1,\lambda)$ design coincides with the classical
notion of a $t$-$(v,k,\lambda)$ design.
Huber \cite{Huber:2010} proved the equivalence between splitting $t$-designs and
optimal splitting A-codes.

\begin{theorem}[Huber \cite{Huber:2010}]
There exists a splitting $t$-$(v,k\times c,1)$ design if and only if
there exists an optimal $(t,c)$-splitting A-code for $k$ equiprobable source states,
having $v$ messages and $\binom{v}{t}/c^t\binom{k}{t}$ encoding rules.
\end{theorem}

The {\em necessary divisibility conditions} for the existence of splitting $t$-designs are as
follows.

\begin{proposition}[Huber \cite{Huber:2010}]
\label{necessary}
The necessary conditions for the
existence of a splitting $t$-$(v,k\times c,\lambda)$ design
are
\begin{equation*}
\lambda\binom{v-s}{t-s}\equiv 0\bmod{c^{t-s}\binom{k-s}{t-s}}, \ \ \text{for all $s$, $0\leq s\leq t$.}
\end{equation*}
\end{proposition}

Sometimes, the points of a splitting $t$-design $(X,\A)$ can be identified with the elements of
an additive group $\Gamma$, so that $X=\Gamma$. If the set of blocks $\A$ can be generated by
a set $\B\subseteq\A$, that is,
\begin{equation*}
\A = \cup_{B\in\B}\{ B + g: g\in \Gamma \},
\end{equation*}
then $\B$ is called a {\em set of base blocks} of $(X,\A)$.

\begin{example}
\label{151}
Let $X=\mathbb{Z}_{151}$. The set containing the single array
\begin{equation*}
A = \left(
\begin{array}{ccccc}
0 & 1 & 2 & 3 & 4 \\
5 & 13 & 59 & 105 & 118 \\
28 & 67 & 73 & 112 & 134
\end{array}
\right)
\end{equation*}
as a base block, generates the set of blocks $\A$ for
a splitting $2$-$(151,3\times 5,1)$ design $(X,\A)$.
\end{example}

Our constructions for splitting $t$-designs also rely on group divisible designs (GDD).
Let $t$, $k$, and $v$ be nonnegative integers.
A {\em group divisible $t$-design of order $v$ and block size $k$},
denoted GDD$(t,k,v)$, is a triple $(X,\G,\A)$ satisfying the following
properties:
\begin{enumerate}[(i)]
\item $X$ is a set of $v$ elements, called {\em points};
\item $\G=\{G_1,\ldots,G_s\}$ is a partition of $X$ into subsets, called {\em groups};
\item $\A\subseteq \binom{X}{k}$, whose elements are called {\em blocks},
such that each $A\in\A$ intersects any group $G\in\G$ in at most one point;
\item every $T\in\binom{X}{t}$ containing at most one point from each group is contained
in exactly one block.
\end{enumerate}
The {\em type} of a GDD$(t,k,v)$ $(X,\G,\A)$ is the multiset $[|G| : G\in \G]$. We use the
exponential notation to describe the type of a GDD: a GDD of type $g_1^{n_1}\cdots g_s^{n_s}$
is a GDD where there are exactly $n_i$ groups of size $g_i$, $1\leq i\leq s$.

We require the following result.

\begin{theorem}[Hanani \cite{Hanani:1975}, Brouwer {\em et al.} \cite{Brouweretal:1977}, Mills \cite{Mills:1990}, Ji \cite{Ji:2009}]\hfill
\label{GDDresults}
\begin{enumerate}[(i)]
\item There exists a {\rm GDD}$(2,3,gn)$ of type $g^n$ if and only if $n\geq 3$,
$(n-1)g\equiv 0\bmod{2}$, and $n(n-1)g^2\equiv 0\bmod{6}$.
\item There exists a {\rm GDD}$(2,4,gn)$ of type $g^n$ if and only if $n\geq 4$, 
$(n-1)g\equiv 0\bmod{3}$, and $n(n-1)g^2\equiv 0\bmod{12}$, with the exception of
$(g,n)\in\{(2,4)$, $(6,4)\}$.
\item For $n>3$, $n\not=5$, a {\rm GDD}$(3,4,gn)$ of type $g^n$ exists if and only if
$gn\equiv 0\bmod{2}$ and $(n-1)(n-2)g\equiv 0\bmod{3}$. A
{\rm GDD}$(3,4,5g)$ of type $g^5$ exists when $g\equiv 0\bmod{2}$, $g\not=2$, and
$g\not\equiv 10, 26\bmod{48}$.
\end{enumerate}
\end{theorem}

Analogous to splitting $t$-designs, a ``splitting'' version of a GDD can be defined.
This has been done by Wang \cite{Wang:2006} for $t=2$. Here, we extend it to general $t$.
A {\em splitting group divisible $t$-design}, denoted splitting GDD$(t, k\times c, v)$,
is a triple $(X,\G,\A)$ satisfying the following properties:
\begin{enumerate}[(i)]
\item $X$ is a set of $v$ elements, called {\em points};
\item $G=\{G_1,\ldots,G_s\}$ is a partition of $X$ into subsets, called {\em groups};
\item $\A$ is a set of $k\times c$ arrays, called {\em blocks}, with entries from $X$, such that
each point of $X$ occurs at most once in each block;
\item for every $\{x_i: 1\leq i\leq t\}\in\binom{X}{t}$ containing at most one point from each group,
there is exactly one block in which $x_i$, $1\leq i\leq t$, occur in $t$ different rows.
\end{enumerate}
The type of a splitting GDD is defined in a fashion similar to that for a GDD.

Splitting GDDs play an important role in the recursive constructions of splitting designs.
The following is a straightforward extension of Wilson's Fundamental Construction for GDDs
\cite{Wilson:1972a,Wilson:1972b}
to splitting GDDs.

\begin{theorem}[Fundamental Construction]
\label{FC}
Let $(X,\G,\A)$ be a {\rm GDD}$(t,k,v)$. Suppose that for each block $A\in\A$, there exists
a splitting {\rm GDD}$(t,k'\times c,kc)$ of type $c^k$, $(X_A,\G_A,\B_A)$, where
\begin{align*}
X_A &= A \times \{1,\ldots,c\}, \\
\G_A &= \{\{x\}\times\{1,\ldots,c\}: x\in A\},
\end{align*}
then there exists a splitting {\rm GDD}$(t,k'\times c, vc)$ of type $[ c|G|: G\in\G]$
$(X',\G',\A')$, where
\begin{align*}
X' &= X\times \{1,\ldots,c\}, \\
\G' &= \{G\times\{1,\ldots,c\} :G\in\G\},\\
\A'&= \cup_{A\in\A} \B_A.
\end{align*}
\end{theorem}

Since the trivial splitting GDD$(t,k\times c,kc)$ of type $c^k$ (containing only one block)
always exists for any $t$, $k$, and $c$, we have the following.

\begin{corollary}
\label{multiply}
If there exists a {\rm GDD}$(t,k,v)$ of type $g_1^{n_1}\ldots g_s^{n_s}$, 
then there exists a splitting
{\rm GDD}$(t,k\times c,vc)$ of type $(cg_1)^{n_1}\ldots (cg_s)^{n_s}$.
\end{corollary}

As shown by Ge {\em et al.} \cite{Geetal:2005},
we can also fill in the groups of a splitting GDD with a splitting 2-design
to obtain new splitting 2-designs.

\begin{proposition}[Filling-In Groups]
\label{fillingroups}
Let $(X,\G,\A)$ be a splitting {\rm GDD}$(2,k\times c,v)$. If for each $G\in\G$,
there exists a splitting $2$-$(|G|+1,k\times c,1)$ design, then there exists a
splitting $2$-$(v+1,k\times c,1)$ design.
\end{proposition}

\subsection{State of Affairs}

The following theorem summarizes the state of knowledge on the existence of splitting
$t$-designs with $\lambda=1$.

\begin{theorem}[Du \cite{Du:2004}, Ge {\em et al.} \cite{Geetal:2005}, Wang \cite{Wang:2006}, Wang and Su \cite{WangSu:2010}]
\label{state}
The necessary divisibility conditions (of Proposition \ref{necessary}) are also sufficient
for the existence of a splitting $2$-$(v,k\times c,1)$ design when
\begin{enumerate}[(i)]
\item $(k,c)=(2,2n)$, for any positive integer $n$;
\item $(k,c)=(2,3)$, except for $v=10$;
\item $(k,c)=(3,2)$, except for $v=9$;
\item $(k,c)=(3,3)$, with the possible exception of $v=55$;
\item $(k,c)=(4,2)$, with the possible exception of $v\in\{49,385\}$.
\end{enumerate}
In addition, there exists a $2$-$(v,3\times 4,1)$ design for all $v\equiv 1\bmod 96$.
\end{theorem}

\section{Nonexistence and Asymptotic Existence}

Let $\lambda$ be a positive integer. The complete (loopless) multigraph on $v$ vertices,
denoted $\lambda K_v$,
is the graph where every pair of distinct vertices is connected by $\lambda$ edges.
Let $G$ be a simple graph without isolated vertices. A {\em $G$-design of order $v$ and index
$\lambda$} is a partition of edge set of $\lambda K_v$ into subgraphs, each of which
is isomorphic to $G$. If $e(G)$ denotes the number of edges in $G$, and $d(G)$ denotes
the greatest
common divisor of the degrees of vertices in $G$, then simple counting shows that the conditions
\begin{enumerate}[(i)]
\item $\lambda v(v-1)\equiv 0\bmod{2e(G)}$, and
\item $\lambda (v-1)\equiv 0\bmod{d(G)}$
\end{enumerate}
are necessary for the existence of a $G$-design of order $v$ and index $\lambda$.
A celebrated result of Wilson \cite{Wilson:1976} states that these necessary conditions are
also asymptotically sufficient.

\begin{theorem}[Wilson \cite{Wilson:1976}]
\label{Wilsonasymp}
Let $G$ be a simple graph without isolated vertices. Then there exists a constant $v_0$
depending only on $G$ and $\lambda$ such that a $G$-design of order $v$ and index
$\lambda$ exist for all $v\geq v_0$ satisfying $\lambda v(v-1)\equiv 0\bmod{2e(G)}$
and $\lambda(v-1)\equiv 0\bmod{d(G)}$.
\end{theorem}

Let $K_{k\times c}$ denote the complete $k$-partite graph, with each part having
$c$ vertices. A splitting $2$-$(v,k\times c,\lambda)$ design $(X,\A)$ is equivalent to a
$K_{k\times c}$-design of order $v$ and index $\lambda$ through the following
correspondence:
\begin{enumerate}[(i)]
\item a point in $X$ corresponds to a vertex in $\lambda K_v$,
\item a block $A\in\A$ corresponds to the complete $k$-partite graph, where the
$i$-th part contains $c$ vertices corresponding to the $c$ entries in row $i$ of $A$, $1\leq i\leq k$.
\end{enumerate}

Applying Theorem \ref{Wilsonasymp} with $G=K_{k\times c}$ then gives the following result
on the asymptotic existence of splitting 2-designs.

\begin{corollary}
There exists a constant $v_0$ depending only on $k$, $c$, and $\lambda$, such that a
splitting $2$-$(v,k\times c,\lambda)$ design exists for all $v\geq v_0$ satisfying
$\lambda v(v-1)\equiv 0\bmod{c^2k(k-1)}$ and
$\lambda(v-1)\equiv 0\bmod{c(k-1)}$.
\end{corollary}

We end this section with a nonexistence result. Huang \cite{Huang:1991} has shown
that the number of complete $k$-partite graphs required to partition the edge set
of $K_v$ is at least $\left\lceil (v-1)/(k-1)\right\rceil$. This has the following consequence.

\begin{proposition}
\label{nonexistence}
There does not exist a splitting $2$-$((k-1)c^2+1,k\times c,1)$ design, for all
$k,c\geq 2$.
\end{proposition}

\begin{proof}
Suppose a splitting $2$-$((k-1)c^2+1,k\times c,1)$ design exists. The
number of blocks in this splitting 2-design is $((k-1)c^2+1)/k$. This would mean that we
can partition the edge set of $K_{(k-1)c^2+1}$ into $((k-1)c^2+1)/k$ complete
$k$-partite subgraphs. This is impossible by Huang's result, since
$\left\lceil (k-1)c^2/(k-1)\right\rceil=c^2>((k-1)c^2+1)/k$.
\end{proof}

The definite exceptions in Theorem \ref{state} are special cases of Proposition \ref{nonexistence}.

\section{Splitting 2-Designs}

In this section, we establish the
existence of an infinite family of splitting $2$-$(v,3\times 5,1)$ designs,
and remove $v=385$ as a possible exception from Theorem \ref{state}(v).

\begin{proposition}
There exists a splitting $2$-$(v,3\times 5,1)$ design for all $v\equiv 1\bmod{150}$, except
possibly when $v=301$.
\end{proposition}

\begin{proof}
A splitting $2$-$(151,3\times 5,1)$ design is exhibited in Example \ref{151}, so let $v\geq 451$.
Write $v=150m+1$, for some integer $m\geq 3$. A GDD$(2,\{3\},30m)$ of type $30^m$
exists by Theorem \ref{GDDresults}(i). Apply Corollary \ref{multiply} to obtain a splitting
GDD$(2,3\times 5,150m)$ of type $150^m$. Now fill in the groups of this splitting GDD with
a splitting $2$-$(151,3\times 5,1)$ design (which has been constructed in Example \ref{151})
to obtain a splitting $2$-$(150k+1,3\times 5,1)$ design.
\end{proof}


\begin{proposition}
There exists a splitting $2$-$(385,4\times 2,1)$ design.
\end{proposition}

\begin{proof}
A GDD$(2,\{4\},192)$ of type $48^4$ exists by Theorem \ref{GDDresults}(ii). Apply 
Corollary \ref{multiply} to obtain a splitting GDD$(2,4\times 2,384)$ of type $96^4$.
Now fill in the groups of this splitting GDD with a splitting $2$-$(97,4\times 2,1)$ design
(which exists by Theorem \ref{state}) to obtain a splitting $2$-$(385,4\times 2,1)$ design.
\end{proof}

\section{Splitting 3-Designs}

In this section, we establish the existence of the first known infinite family of
splitting 3-designs with $c>1$.

Let $t$, $k$, and $v$ be nonnegative integers. A {\em $(t,k)$ candelabra system}
of order $v$ is a quadruple $(X,S,\G,\A)$ that satisfies the following properties:
\begin{enumerate}[(i)]
\item $X$ is a set of $v$ elements, called {\em points};
\item $S\subseteq X$, called the {\em stem};
\item $\G=\{G_1,\ldots,G_m\}$ is a partition of $X\setminus S$ (elements of $\G$ are called
{\em groups});
\item $\A\subseteq\binom{X}{k}$, whose elements are called {\em blocks};
\item every $T\in\binom{X}{t}$ with $|T\cap(S\cup G_i)|<t$ for all $i$, is contained in
a block in $\A$.
\end{enumerate}
The {\em type} of a $(t,k)$ candelabra  system $(X,S,\G,\A)$ is the multiset
$[|G|:G\in\G]$. A $(t,k)$ candelabra system of type $g_1^{n_1}\cdots g_r^{n_r}$ with
a stem of size $s$ is denoted $(t,k)$-CS$(g_1^{n_1}\cdots g_r^{n_r}:s)$.

Here, we introduce the notion of splitting candelabra systems.

A {\em splitting $(t,k\times c)$ candelabra system} of order $v$ is a quadruple
$(X,S,\G,\A)$ that satisfies the following properties:
\begin{enumerate}[(i)]
\item $X$ is a set of $v$ elements, called {\em points};
\item $S\subseteq X$, called the {\em stem};
\item $\G=\{G_1,\ldots,G_m\}$ is a partition of $X\setminus S$ (elements of $\G$ are called
{\em groups});
\item $\A$ is a set of $k\times c$ arrays, called {\em blocks}, with entries from $X$, such
that each point of $X$ occurs at most once in each block;
\item for every $\{x_i:1\leq i\leq t\}\in\binom{X}{t}$ with $|T\cap(S\cup G_i)|<t$ for all $i$, there is
exactly one block in which $x_i$, $1\leq i\leq t$, occur in $t$ different rows.
\end{enumerate}
We use the same notation for splitting $(t,k)$ candelabra systems as those for $(t,k)$
candelabra systems.

The following theorem is an extension of Hartman's Fundamental Construction
\cite{Hartman:1994} from $(3,k)$ candelabra systems to splitting $(3,k)$ candelabra systems.

\begin{theorem}
\label{FC3}
If there exist a $(3,k)$-{\rm CS}$(g_1^{n_1}\cdots g_r^{n_r}:s)$, 
a splitting $(3,k'\times c)$-{\rm CS}$(m^{k-1}:a)$, and
a splitting {\rm GDD}$(3,k'\times c,mk)$ of type $m^k$, then
there exists a splitting $(3,k'\times c)$-{\rm CS}$((g_1m)^{n_1}\cdots (g_rm)^{n_r}:m(s-1)+a)$.
\end{theorem}

\begin{proof}
Let $(X,S,\G,\A)$ be a $(3,k)$-CS$(g_1^{n_1}\cdots g_r^{n_r}:s)$, and let 
$\infty$ be a distinguished
point in $S$. For $Y\subseteq X$, define the set of points
\begin{equation*}
P(Y)= ((Y\setminus\{\infty\}) \times \mathbb{Z}_m) \cup (\{\infty\}\times\mathbb{Z}_a).
\end{equation*}
Further define
\begin{align*}
S' &= ((S\setminus\{\infty\}) \times \mathbb{Z}_m) \cup (\{\infty\}\times\mathbb{Z}_a), \\
\G' &= \{G\times\mathbb{Z}_m:G\in\G\}.
\end{align*}

For each $A\in\A$ containing the point $\infty$, let
\begin{equation*}
(P(A),\{\infty\}\times\mathbb{Z}_a,\{\{x\}\times\mathbb{Z}_m:x\in A\setminus\{\infty\}\},\B_A)
\end{equation*}
be a splitting $(3,k'\times c)$-CS$(m^{k-1}:a)$, 
and for each $A\in\A$ not containing the point $\infty$,
let 
\begin{equation*}
(A\times\mathbb{Z}_m,\{\{x\}\times\mathbb{Z}_m: x\in A\}, \C_A)
\end{equation*}
be a splitting GDD$(3,k'\times c,3m)$ of type $m^k$.

It is easy to check that $(P(X),S',\G',\A')$, where
\begin{equation*}
\A' = \left(\bigcup_{A\in\A:\infty\in A}\B_A\right) \cup \left(\bigcup_{A\in\A:\infty\not\in A}\C_A\right),
\end{equation*}
is the required splitting $(3,k'\times c)$-CS$((g_1m)^{n_1}\cdots (g_rm)^{n_r}:m(s-1)+a)$.
\end{proof}

We can also
fill in the groups of a splitting candelabra system by splitting 3-designs to obtain
larger splitting 3-designs.

\begin{proposition}
\label{fillin3}
If there exists a splitting $(3,k\times c)$-{\rm CS}$(g_1^{n_1}\cdots g_r^{n_r}:s)$, where $s\leq 2$,
and there exists a splitting $3$-$(g_i+s,k\times c,1)$ design for each $i$, $1\leq i\leq r$,
then there exists a splitting $3$-$(s+\sum_{i=1}^r g_in_i, k\times c,1)$ design.
\end{proposition}

\begin{proof}
Let $(X,S,\G,\A)$ be a splitting $(3,k\times c)$-CS$(g_1^{n_1}\cdots g_r^{n_r}:s)$, where
$s\leq 2$. For each $G\in\G$, let $(G\cup S,\B_G)$ be a splitting 3-$(|G|+s,k\times c,1)$ design.
Then $(X,\A\cup (\cup_{G\in\G}\B_G))$ is the required splitting 
$3$-$(s+\sum_{i=1}^r g_in_i, k\times c,1)$ design.
\end{proof}

To apply Theorem \ref{FC3} and Proposition \ref{fillin3}, we require some splitting candelabra
systems to start with. 

\begin{lemma}
\label{small}
There exist a splitting $(3,3\times 2)$-{\rm CS}$(8^2:0)$ and a splitting
$(3,3\times 2)$-{\rm CS}$(8^2:2)$.
\end{lemma}

\begin{proof}
Let $X=\mathbb{Z}_{16}$ and $\G=\{ \{2i+j:0\leq i\leq 7\} : j \in\{0,1\} \}$.
Let 
\begin{align*}
\B = \left\{
\begin{pmatrix}
 0 & 4 \\
 6 & 9 \\
7 & 11
\end{pmatrix},\right.
\begin{pmatrix}
 0 & 14 \\
 1 & 4 \\
11 & 13
\end{pmatrix},
&\begin{pmatrix}
 0 & 5 \\
 8 & 10 \\
13 & 15
\end{pmatrix},
\begin{pmatrix}
 0 & 2 \\
 4 & 1 \\
7 & 15
\end{pmatrix}, \\
&\left.\begin{pmatrix}
 0 & 13 \\
 1 & 15 \\
2 & 12
\end{pmatrix},
\begin{pmatrix}
 0 & 13 \\
 1 & 9 \\
4 & 6
\end{pmatrix},
\begin{pmatrix}
 0 & 6 \\
 9 & 7 \\
14 & 15
\end{pmatrix}
\right\}.
\end{align*}
Then $(X,\G,\varnothing,\A)$, where
$\A=\cup_{B\in\B} \{ B+2i\bmod{16} : 0\leq i <8\}$,
is a splitting $(3,3\times 2)$-CS$(8^2:0)$.

Now let $S=\{x,y\}$ be such that $S\cap X=\varnothing$, and let
\begin{equation*}
\C = \left\{
\left(
\begin{array}{cc}
x & y \\
2i & 2i+2 \\
2j+1 & 2j+3
\end{array}
\right) : i,j\in\{0,2,4,6\}
\right\}.
\end{equation*}
Then $(X\cup\{x,y\},S,\G,\A\cup\C)$ is a splitting $(3,3\times 2)$-CS$(8^2:2)$.
\end{proof}

We now establish an infinite family of splitting 3-designs.

\begin{theorem}
A splitting $3$-$(v,3\times 2,1)$ design exists if and only if $v\equiv 2\bmod{8}$.
\end{theorem}

\begin{proof}
Necessity of the condition $v\equiv 2\bmod{8}$ follows from Proposition \ref{necessary}.

Huber \cite{Huber:2010} has shown the existence of a splitting $3$-$(10,3\times 2,1)$ design,
so we consider $v>10$. Write $v=8m+2$, for some $m\geq 2$. Let $X$ be a set
of $m+1$ points, containing $\infty$ as a distinguished point. It is easy to verify that
$(X,\{\infty\},\{\{x\}: x\in X\setminus\{\infty\}\},\binom{X}{3})$ is a $(3,3)$-CS$(1^m:1)$.
Apply Theorem \ref{FC3} with a splitting $(3,3\times 2)$-CS$(8^2:2)$ (which exists by
Lemma \ref{small})
and a splitting GDD$(3,3\times 2,24)$ of type $8^3$ (whose existence is implied by
the trivial GDD$(3,3,12)$ of type $4^3$ and Corollary \ref{multiply}) to obtain a 
splitting $(3,3\times 2)$-CS$(8^m:2)$. Now apply Proposition \ref{fillin3} to this
splitting $(3,3\times 2)$-CS$(8^m:2)$ with a 
splitting $3$-$(10,3\times 2,1)$ design to obtain a splitting $3$-$(8m+2,3\times 2,1)$ design.
\end{proof}

\section{Conclusion}

Determining the existence of optimal $c$-splitting authentication codes
with $k$ source states
that are $(t-1)$-fold secure against spoofing is a difficult problem, when $k$, $c$ 
and $t$ are large.
New constructions, both direct and recursive, need to be developed in order to
make further progress on the problem.

\section*{Acknowledgment}

The authors would like to thank Gennian Ge and Alan Ling for helpful discussions.


\providecommand{\bysame}{\leavevmode\hbox to3em{\hrulefill}\thinspace}
\providecommand{\MR}{\relax\ifhmode\unskip\space\fi MR }
\providecommand{\MRhref}[2]{%
  \href{http://www.ams.org/mathscinet-getitem?mr=#1}{#2}
}
\providecommand{\href}[2]{#2}

\medskip
Received xxxx 20xx; revised xxxx 20xx.
\medskip


\begin{thebibliography}{10}

\bibitem{Brouweretal:1977}
A.~E. Brouwer, A.~Schrijver, and H.~Hanani, \emph{Group divisible designs with
  block-size four}, Discrete Math. \textbf{20} (1977), no.~1, 1--10.

\bibitem{Du:2004}
B.~Du, \emph{Splitting balanced incomplete block designs with block size
  {$3\times 2$}}, J. Combin. Des. \textbf{12} (2004), no.~6, 404--420.

\bibitem{Geetal:2005}
G.~Ge, Y.~Miao, and L.~Wang, \emph{Combinatorial constructions for optimal
  splitting authentication codes}, SIAM J. Discrete Math. \textbf{18} (2005),
  no.~4, 663--678.

\bibitem{Hanani:1975}
Haim Hanani, \emph{Balanced incomplete block designs and related designs},
  Discrete Math. \textbf{11} (1975), 255--369.

\bibitem{Hartman:1994}
A.~Hartman, \emph{The fundamental construction for {$3$}-designs}, Discrete
  Math. \textbf{124} (1994), no.~1-3, 107--132.

\bibitem{Huang:1991}
Q.~X. Huang, \emph{On the decomposition of {$K\sb n$} into complete
  {$m$}-partite graphs}, J. Graph Theory \textbf{15} (1991), no.~1, 1--6.

\bibitem{Huber:2010}
M.~Huber, \emph{Combinatorial bounds and characterizations of splitting
  authentication codes}, Cryptogr. Commun. \textbf{2} (2010), no.~2, 173--185.

\bibitem{Ji:2009}
L.~Ji, \emph{An improvement on {$H$} design}, J. Combin. Des. \textbf{17}
  (2009), no.~1, 25--35.

\bibitem{KurosawaObana:2001}
K.~Kurosawa and S.~Obana, \emph{Combinatorial bounds on authentication codes
  with arbitration}, Des. Codes Cryptogr. \textbf{22} (2001), no.~3, 265--281.

\bibitem{Massey:1986}
J.~L. Massey, \emph{Cryptography, a selective survey}, Digital Communications
  '85: Proceedings of the Second Tirrenia International Workshop on Digital
  Communications (E.~Biglieri and G.~Prati, eds.), Elsevier Science Pub. Co.,
  1986, pp.~3--25.

\bibitem{Mills:1990}
W.~H. Mills, \emph{On the existence of {$H$} designs}, Proceedings of the
  {T}wenty-first {S}outheastern {C}onference on {C}ombinatorics, {G}raph
  {T}heory, and {C}omputing ({B}oca {R}aton, {FL}, 1990), vol.~79, 1990,
  pp.~129--141.

\bibitem{Ogataetal:2004}
W.~Ogata, K.~Kurosawa, D.~R. Stinson, and H.~Saido, \emph{New combinatorial
  designs and their applications to authentication codes and secret sharing
  schemes}, Discrete Math. \textbf{279} (2004), no.~1-3, 383--405.

\bibitem{Simmons:1982}
G.~J. Simmons, \emph{A game theory model of digital message authentication},
  Congr. Numer. \textbf{34} (1982), 413--424.

\bibitem{Simmons:1984}
\bysame, \emph{Message authentication: a game on hypergraphs}, Congr. Numer.
  \textbf{45} (1984), 161--192.

\bibitem{Simmons:1985}
\bysame, \emph{Authentication theory/coding theory}, Advances in Cryptology --
  {CRYPTO} '84 (G.~R. Blakely and D.~Chaum, eds.), Lecture Notes in Comput.
  Sci., vol. 196, Springer-Verlag, 1985, pp.~411--432.

\bibitem{Simmons:1987}
\bysame, \emph{Message authentication with arbitration of transmitter/receiver
  disputes}, Advances in Cryptology -- {EUROCRYPT} '87, Lecture Notes in
  Comput. Sci., vol. 304, Springer-Verlag, 1987, pp.~151--165.

\bibitem{Simmons:1990}
\bysame, \emph{A {C}artesian product construction for unconditionally secure
  authentication codes that permit arbitration}, J. Cryptology \textbf{2}
  (1990), no.~2, 77--104.

\bibitem{Simmons:1992}
\bysame, \emph{A survey of information authentication}, Contemporary Cryptology
  --- The Science of Information Integrity (G.~J. Simmons, ed.), IEEE Press,
  1992, pp.~379--419.


\bibitem{Wang:2006}
J.~Wang, \emph{A new class of optimal 3-splitting authentication codes}, Des.
  Codes Cryptogr. \textbf{38} (2006), no.~3, 373--381.

\bibitem{WangSu:2010}
J.~Wang and R.~Su, \emph{Further results on the existence of splitting {BIBD}s
  and application to authentication codes}, Acta Appl. Math. \textbf{109}
  (2010), no.~3, 791--803.

\bibitem{Wilson:1972a}
R.~M. Wilson, \emph{An existence theory for pairwise balanced designs. {I}.
  {C}omposition theorems and morphisms}, J. Combin. Theory Ser. A \textbf{13}
  (1972), 220--245.

\bibitem{Wilson:1972b}
\bysame, \emph{An existence theory for pairwise balanced designs. {II}. {T}he
  structure of {PBD}-closed sets and the existence conjectures}, J. Combin.
  Theory Ser. A \textbf{13} (1972), 246--273.

\bibitem{Wilson:1976}
\bysame, \emph{Decompositions of complete graphs into subgraphs isomorphic to a
  given graph}, Proceedings of the {F}ifth {B}ritish {C}ombinatorial
  {C}onference ({U}niv. {A}berdeen, {A}berdeen, 1975), Congressus Numerantium,
  No. XV, Utilitas Math., Winnipeg, Man., 1976, pp.~647--659.

\end{thebibliography}
\end{document}